\documentclass[%
11pt,
reprint,
onecolumn,
tightenlines,
superscriptaddress,
preprintnumbers,
nofootinbib,
amsmath,amssymb,amsthm,
physrev,
eqsecnum,tikz,
]{revtex4-2}

\usepackage{isomath}
\usepackage{amsmath,amsthm}
\usepackage{amsbsy}
\usepackage{amssymb}
\usepackage{amscd}
\usepackage{amsfonts}
\usepackage{stmaryrd}
\usepackage{siunitx}
\usepackage{euscript}
\usepackage[utf8]{inputenc}
\usepackage[T1]{fontenc}
\usepackage{newtxtext} 
\everymath{\displaystyle}
\usepackage{exscale}
\usepackage{microtype}
\usepackage{booktabs}
\usepackage{algorithm}
\usepackage{algpseudocode}
\usepackage{hyperref}

\usepackage{graphicx}
\usepackage{boxedminipage}
\usepackage{calc}
\usepackage[dvipsnames]{xcolor}
\graphicspath{ {media/} }
\usepackage[caption=false,justification=raggedright]{subfig}

\usepackage{setspace}
\usepackage{enumitem}
\setitemize{noitemsep,topsep=0pt,parsep=0pt,partopsep=0pt}
\setenumerate{noitemsep,topsep=0pt,parsep=0pt,partopsep=0pt}
\setdescription{noitemsep,topsep=0pt,parsep=0pt,partopsep=0pt}

\usepackage[colorinlistoftodos, color=green!40]{todonotes}
\setuptodonotes{inline,caption={}}

\usepackage{soul} 
\usepackage[normalem]{ulem}

\usepackage{orcidlink}
\usepackage{siunitx}
\usepackage[small]{titlesec}

\titlespacing*{\section}{0pt}{12pt plus 4pt minus 2pt}{2pt plus 2pt minus 2pt}
\titlespacing*{\subsection}{0pt}{12pt plus 4pt minus 2pt}{2pt plus 2pt minus 2pt}
\titlespacing*{\subsubsection}{0pt}{12pt plus 4pt minus 2pt}{2pt plus 2pt minus 2pt}
\titlespacing*{\paragraph}{0pt}{12pt plus 4pt minus 2pt}{2pt plus 2pt minus 2pt}

\makeatletter

    \renewcommand*{\p@subsection}{}
    
    \renewcommand*{\p@subsubsection}{}
\makeatother

\usepackage{upgreek}

\newtheorem{theorem}{Theorem}[section]
\newtheorem{lemma}{Lemma}[section]

\theoremstyle{definition}

\AtEndEnvironment{definition}{\null\hfill\qedsymbol}
\newtheorem{remark}{Remark}[section]
\AtEndEnvironment{remark}{\null\hfill\qedsymbol}

\AtEndEnvironment{example}{\null\hfill\qedsymbol}

\AtEndEnvironment{assumption}{\null\hfill\qedsymbol}

\DeclareMathOperator{\divergence}{div}

\newcommand{\dm}{\ \mathrm{d}}

\usepackage{times}
\usepackage{amsfonts}

\usepackage{amsopn}

\usepackage{amsmath}
\usepackage{bbm}
\usepackage{stmaryrd}
\usepackage[rightcaption]{sidecap}
\usepackage{accents}
\usepackage{footnote}

\newcommand{\Scr}[1]{\mathbb{#1}}
\newcommand{\scr}[1]{\mathcal{#1}}

\newcommand{\ve}[1]{\boldsymbol{#1}}

\newcommand{\utilde}[1]{\underaccent{\tilde}{#1}}

\newtheorem{exmp}{Example}[section]
\newcommand{\beol}[1]{\begin{equation}
#1
\end{equation}}
\newcommand{\beolnn}[1]{\begin{equation*}
#1
\end{equation*}}
\newcommand{\bml}[1]{\begin{equation}
\begin{split}
#1
\end{split}
\end{equation}}
\newcommand{\bmlnn}[1]{\begin{equation*}
\begin{split}
#1
\end{split}
\end{equation*}}

\def\lr{\mbox{\begin{picture}(7,10)
\put(1,0){\line(1,0){5}}
\put(1,0){\line(1,2){5}}
\put(6,0){\line(0,1){10}}
\end{picture}
}}

\newcommand{\rmrk}[1]{\begin{remark}
#1
\end{remark}}

\begin{document}
	
	\preprint{SIAM Multiscale Modeling and Simulation 24(1):162-186, 2026 (DOI: 10.1137/25M1747907)}
	
	\title{Two-Scale Analysis of the Electrostatics of Dielectric Crystals: Emergence of Polarization Density and Boundary Charges}
	
	\author{Shoham Sen}
	\affiliation{Department of Civil and Environmental Engineering, Carnegie Mellon University}
	
	\author{Yang Wang}
	\affiliation{Pittsburgh Supercomputing Center}
	
	\author{Timothy Breitzman}
	\affiliation{Air Force Research Laboratory}
	
	\author{Kaushik Dayal}
	\email{Kaushik.Dayal@cmu.edu}
	\affiliation{Department of Civil and Environmental Engineering, Carnegie Mellon University}
	\affiliation{Center for Nonlinear Analysis, Department of Mathematical Sciences, Carnegie Mellon University}
	\affiliation{Department of Mechanical Engineering, Carnegie Mellon University}
	
	\date{\today}

	
	\begin{abstract}
		Ionic crystals, such as solid electrolytes and complex oxides, are central to modern technologies for energy storage, sensing, actuation, and other functional applications. An important fundamental issue in the atomic and quantum-scale modeling of these materials is defining the macroscopic polarization. In a periodic crystal, the usual definition of the polarization as the first moment of the charge density in a unit cell is found to depend qualitatively – allowing even a change in the sign! – and quantitatively on the choice of unit cell.
		
		We examine this issue using a rigorous approach based on the framework of 2-scale convergence. By examining the continuum limit of when the lattice spacing is much smaller than the characteristic dimensions of the body, we show that the 2-scale  limit provides both a bulk polarization as well as a surface charge density supported on the boundary of the body.    
		Further, different choices of the periodic unit cell of the body lead to correspondingly different partial unit cells at the boundary; these choices give to different bulk polarization and surface charges but compensate such that the electric field and energy are independent of the choice of unit cell.
	\end{abstract}
	
	\maketitle


\section{Introduction}

Crystalline ionic materials, such as dielectrics, solid electrolytes, and ferroelectrics, are essential for engineering applications such as energy storage, sensing, actuation, and numerous functional settings.
Continuum models at scales of engineering relevance typically use the continuum dipole per unit volume --- the polarization --- as a multiscale mediator that captures key aspects of the atomic-scale charge distribution
\cite{james1994internal,xiao2005influence,tadmor-EffHamil,cicalese2009discrete,bach2020discrete,alicandro2008continuum,sharp1994electrostatic,grasinger2020statistical,grasinger2021architected,grasinger2022statistical,grasinger2021nonlinear,khandagale2024statistical,khandagale2025nonlocal}. The polarization field consequently is a central quantity of interest in various successful macroscopic continuum models of dielectric and functional materials, e.g., piezoelectricity \cite{haeni2004room,tagantsev1986piezoelectricity,tagantsev2010domains,rahmati2019nonlinear,deng2017continuum}, flexoelectricity \cite{zubko2013flexoelectric,abdollahi2015revisiting,abdollahi2014computational,abdollahi2019converse,abdollahi2015fracture,liu2013flexoelectricity,krichen2016flexoelectricity,grasinger2021flexoelectricity,liu2013flexoelectricity,ahmadpoor2013apparent,steigmann2018mechanics,torbati2022coupling, liu2018emergent}, ferroelectricity and electromechanics \cite{Irene1,yang2011completely,yang2011effect,yang2012free,yang2012influence,yang2012microstructure,yen2008study,shu2008constrained,katusele2025exploiting,katusele2025soft}, and analogous models in micromagnetism \cite{james1990frustration,james1998magnetostriction,james2000martensitic,desimone2002constrained,tickle1999ferromagnetic,benesova2018existence,wang2001gauss,garcia2003accurate,bosse2007microdomain}.

An important feature of the polarization is that the continuum value depends on the choice of lattice unit cell. In typical formal calculations performed in \cite{sen2024nonuniqueness, Resta-Vanderbilt, marshall2014atomistic, spaldin2012beginner}, it can be shown that with the appropriate change in the choice of unit cells, the sign of the polarization can be made negative. For example, consider an illustrative one-dimensional case of an infinite line of alternating positive and negative charges. The dipole moment can be defined as a vector connecting a negative charge to its nearest positive neighbor. However, since each negative charge has two equidistant positive neighbors, the choice of reference introduces a sign ambiguity in the dipole moment. Given that polarization serves as a macroscopic descriptor of the electronic properties of a material, its non-uniqueness introduces challenges for the reliability of predictions made using continuum models.

To address this challenge, we apply concepts of two-scale convergence as introduced by \cite{allaire1992homogenization}. The advantage of two-scale convergence lies in its ability to capture specific physical phenomena at the microscopic scale. Unlike weak convergence, which yields an `averaged' function in the limit, two-scale convergence preserves information about the function's local oscillations. Using this framework, we derive two equations: one describing the physics at the macroscopic scale and another for the microscopic scale. We establish a polarization theorem for the dipole from the charge density, which enables us to compute both the homogenized (macroscopic-scale) potential and the corrector (microscopic-scale) potential. Utilizing the property of weak-convergence \cite{murat1997h,allaire1997shape,gustafsson2006note}, we demonstrate the uniqueness of the homogenized limit for the potential.

Mathematical approaches to related problems typically begin with a finite body and then take the limit; such approaches can naturally account for surface effects. For instance, \cite{james1994internal, jha2023discrete, muller2002discrete, schlomerkemper-schmidt,jha2023discrete} take this approach; however, they begin with microscopic dipoles rather than a charge distribution and hence do not have surface effects. Related work in \cite{rosakis2014continuum} with short-range atomic forces shows the emergence of a surface energy. In \cite{jha2023atomic,marshall2014atomistic}, they formally study such limits using distributions of charges. Their findings demonstrate that the coarse-grained description necessitates both a polarization—defined as the dipole moment of the charge within a given periodic unit cell—and the associated surface charge density corresponding to the respective partial unit cell.

A Modern Theory of Polarization \cite{resta2007theory} introduces an alternative definition of polarization through the change in the Berry phase of the wavefunction during an adiabatic process. The central assertion is that, while polarization itself is non-unique, the change in polarization is uniquely defined modulo a polarization quantum. The assumption of adiabaticity implies that eigenstates remain the same throughout the ``change", allowing one to get an expression for the current using the eigenstates at equilibrium. Thus, by reformulating the multiscale mediator from the polarization to the change in polarization, they are able to establish unique values (modulo the polarization quantum). However, the theory relies fundamentally on the independent-electron approximation and the validity of adiabatic evolution. As such, it is not applicable to systems exhibiting strong electronic correlations or undergoing displacive first-order phase transitions, where changes in polarization may occur abruptly and cannot be captured through an adiabatic process  \cite{sen2024nonuniqueness,Sen2022}.

\subsection{Contributions of this Work}

This paper focuses on deriving the limiting form of the electrostatic potential for a solid. By assuming that the material is locally periodic, we extract two-scale equations to describe the behavior. Specifically, we provide a rigorous derivation of the homogenized limit for the electrostatic interaction between nuclei and electrons within a charge distribution defined over a bounded domain $\Omega \subset \mathbb{R}^d$, with particular interest in the case $d=3$. Due to the long-range nature of electrostatic interactions, the sum of the charge density over lattice unit cells generally does not converge \cite{marshall2014atomistic}. As discussed in \cite{marshall2014atomistic}, a net charge neutrality condition in each unit cell is necessary to obtain a well-defined electrostatic potential in the limit. For further details, we refer the reader to Section 3.1 of \cite{marshall2014atomistic}. This work aims to rigorously establish the results presented in \cite{marshall2014atomistic}

Classical homogenization problems involve the homogenization of linear operators with periodically oscillating coefficients and fixed source terms. For the sequence of problems in \eqref{OperatorHomG}, the homogenization process replaces this sequence with a single operator, $L_0$, such that the solutions corresponding to the original operators converge weakly to the solution associated with $L_0$.
\begin{equation}
    L_{\varepsilon}u_{\varepsilon}=f \overset{\varepsilon\downarrow0}{\longrightarrow} L_{0}u_{0}=f, \label{OperatorHomG}
\end{equation}
where the domain of definition of the partial differential equation (PDE) is $\Omega$, and the above PDE's are complemented with proper boundary conditions. In the problem under consideration, the operator is fixed, whereas the forcing term varies.
\begin{equation*}
    Lu_{\varepsilon}=f_{\varepsilon}. 
\end{equation*}

We derive the limiting form of the source term $f_{\varepsilon}$ and obtain a homogenized PDE in the limit. Since the operator $L$ is fixed, the solution can be expressed as the integral of $f_{\varepsilon}$ against the system's Green's function. Given that the Green's function remains unchanged, weak convergence can be employed to determine the limit. This approach was used formally in \cite{marshall2014atomistic}, where the authors work with the free-space Green's function rather than a general Green's function. In this paper, we apply 2-scale analysis to rigorously consider the homogenized limit of the aforementioned PDE over general domains. Additionally, the dipole limit of the electron density provides not only a definition of the polarization but also the surface charge density in the limit. 

The primary result of this paper, presented in Section \ref{WF}, establishes that a sequence of Poisson equations defined for a sequence of charge distributions $\{\tilde{\rho}_l\}:\Omega \to \mathbb{R}$, associated with a scaled lattice $l\scr{L}$ (where $l$ scales the lattice $\scr{L}$), possesses a well-defined limit. By imposing the condition that the charge distribution has zero mean within each unit cell of $l\scr{L}$, the sequence of partial differential equations can be formulated as:
\begin{equation}
\begin{split}
    -\Delta&\Phi_l=\dfrac{\tilde{\rho}_l}{l}\mathbbm{1}_{\Omega},\\
	 &\Phi_l=f~\text{on}~\Gamma_d,\\
	 &\Phi_l(\ve{x})=0 \text{ as }|\ve{x}|\to\infty,\\
     &\Phi_l\in L^2_{loc}(\Scr{R}^3)~~\text{and}~~\nabla\Phi_l\in L^2(\Scr{R}^3),
\end{split}\label{1.2}
\end{equation}
where $\Omega\subset\Scr{R}^3$ is the support of the charge distribution and the potential is specified on $\Gamma_d$ part of the boundary $\partial\Omega$. The last line in \eqref{1.2} follows \cite{tian2012dielectric}, and ensures that the total energy is finite.
The homogenized limit of the above PDE is:
\begin{equation}
\begin{split}
    \Delta_x&\Phi_0(\ve{x})=\mathbbm{1}_{\Omega}\divergence(\ve{p}_0)~~\text{in}~\Scr{R}^3,\\
	 &\Phi_0=f~\text{on}~\Gamma_d,\\
	 &\left\llbracket\dfrac{\partial\Phi_0}{\partial n}\right\rrbracket=\ve{p}_0\cdot \ve{n}+\sigma~\text{on}~\partial\Omega\backslash\Gamma_d,\\
     &\nabla\Phi_l~\text{is bounded uniformly in $L^2(\Omega)$},
\end{split}\label{1.3}
\end{equation}
where $\sigma$ is a surface charge due to partial unit cells and $\ve{p}_0$ is the polarization distribution due to the charge $\{\tilde{\rho}_l\}$. The last condition corresponds to the energy of the sequence being uniformly bounded in $l$.

\subsection{Organization}

The structure of the paper is as follows. In Section \ref{Intuition}, we provide an intuitive overview of the problem, considering a simplified version that parallels the analysis in \cite{marshall2014atomistic}. Readers seeking a more visual explanation of the homogenization process may refer to the corresponding section in \cite{sen2024nonuniqueness}. Section \ref{Notation} introduces the notation that will be utilized throughout the paper to ensure clarity and consistency. In Section \ref{PF}, we formulate the electrostatic equation along with the appropriate boundary conditions. Section \ref{WF} presents the weak two-scale formulation used to derive the homogenized equation in the limit as $l \to 0$. This section yields two independent equations: one describing the homogenized potential $\Phi_0(\ve{x})$ and the other defining the corrector field $\Phi_1(\ve{x},\ve{y})$, both expressed in terms of the charge density $\Tilde{\rho}_l(\ve{x})$. Finally, the proof of the polarization theorem is reserved for Section \ref{2sTm}.

In Section \ref{2sTm}, we establish the weak convergence of zero-mean functions to a bulk dipole term and a surface charge density term. The integral over the domain $\Omega$ is decomposed into a sum over unit cells $l\Box$ (collectively denoted as $\Omega(\Box)$) and a residual term arising from partial unit cells (collectively denoted as $\Omega(\lr)$). The treatment of these components is divided between Section \ref{VT} and Section \ref{SurfCrgDnsty}. In both sections, we rely on the boundedness of the zero-mean charge density within each unit cell and the continuity and differentiability of the test functions. In Section \ref{VT}, we express the bulk term as a Riemann summation over the domain. By leveraging the boundedness of the integrand, we extract a convergent subsequence. In Section \ref{SurfCrgDnsty}, we construct the integral for the surface charge density, showing that it converges to a well-defined surface-integral limit for the residual term. 

In Section \ref{Non-uniqueness}, we address the non-uniqueness of polarization arising from different choices of unit cells, noting that each choice corresponds to a distinct surface charge density in the limit. Despite this inherent non-uniqueness of polarization and surface charge density, we demonstrate that the homogenized potential remains uniquely defined.

\subsection{Heuristics}\label{Intuition}
The aim of this section is to provide motivation for the rigorous calculation presented later by considering a simplified version of problem \eqref{1.2}. Specifically, we set $|\Gamma_d|=0$, allowing the potential to be expressed in terms of the free-space Green's function and the charge distribution (see \cite{marshall2014atomistic}). Additionally, we eliminate the explicit dependence on $l$ in the sequence of charge distributions in \eqref{1.2}, under the assumption that $l$ is sufficiently small. In this simplified setting, the potential can be represented as:
\begin{equation}
 \Phi_l(\ve{x})=\int_{\Omega} G(\ve{x},\ve{x}')\dfrac{\Tilde{\rho}(\ve{x}')}{l} \dm\ve{x}',\label{formal-Int}   
\end{equation}
where $G(\ve{x},\ve{x}')$ is the free space Greens function. The domain $\Omega$ is partitioned into full ($\Box$) and partial unit cells ($\lr$); the former consists of charge-neutral unit cells, while the latter is not charge-neutral. The unit cells are chosen to be unit cubes of the form $[0,1)^3$, with the associated corner mapping ($\hat{\cdot}$) defined as $\ve{x} = x^i\ve{e}_i \mapsto [x^i]\ve{e}_i$, where $[\cdot]$ denotes the floor function. Note that the floor function acts on the coordinates.

Thus the integration can be represented by:
\begin{equation}
    \Phi_l(\ve{x})=\sum_{\Omega(\Box)}\int_{l\Box} G(\ve{x},\ve{x}')\dfrac{\Tilde{\rho}(\ve{x}')}{l} \dm\ve{x}'+\sum_{\Omega(\lr)}\int_{l\lr} G(\ve{x},\ve{x}')\dfrac{\Tilde{\rho}(\ve{x}')}{l} \dm\ve{x}',
\end{equation}
where $\Omega(\Box)$ is the collection of full unit cells in the bulk and $\Omega(\lr)$ is the collection of partial unit cells near the surface.

Next, we substitute $\ve{x}' = \hat{\ve{x}}' + l\ve{y}'$, where $\hat{\ve{x}}$ corresponds to the corner map defined previously. Applying a Taylor expansion to the Green’s function yields:
\begin{equation}
    \begin{split}
\Phi_l(\ve{x})&=\sum_{\hat{\ve{x}}'\in\Omega(\Box)}l^3\int_{\Box} G(\ve{x},\hat{\ve{x}}'+l\ve{y}')\dfrac{\Tilde{\rho}(\hat{\ve{x}}'+l\ve{y}')}{l} \dm\ve{y}'+\sum_{\hat{\ve{x}}'\in\Omega(\lr)}l^3\int_{\lr} G(\ve{x},\hat{\ve{x}}'+l\ve{y}')\dfrac{\Tilde{\rho}(\hat{\ve{x}}'+l\ve{y}')}{l} \dm\ve{y}'\\
    &\approx\sum_{\hat{\ve{x}}'\in\Omega(\Box)}l^3\int_{\Box} \left[G(\ve{x},\hat{\ve{x}}')+l\ve{y}'\cdot\nabla_{\ve{x}'}G(\ve{x},\ve{x}')\right]\dfrac{\Tilde{\rho}(\hat{\ve{x}}'+l\ve{y}')}{l}\dm\ve{y}'+\sum_{\hat{\ve{x}}'\in\Omega(\lr)}l^2\int_{\lr} G(\ve{x},\hat{\ve{x}}')\Tilde{\rho}(\hat{\ve{x}}'+l\ve{y}') \dm\ve{y}',\\
    &=\sum_{\hat{\ve{x}}'\in\Omega(\Box)}l^3 \nabla_{\ve{x}'}G(\ve{x},\ve{x}')\cdot\int_{\Box} \ve{y}'\Tilde{\rho}(\hat{\ve{x}}'+l\ve{y}')\dm\ve{y}'+\sum_{\hat{\ve{x}}'\in\Omega(\lr)}l^2G(\ve{x},\hat{\ve{x}}')\int_{\lr} \Tilde{\rho}(\hat{\ve{x}}'+l\ve{y}') \dm\ve{y}',\\
    &=\sum_{\hat{\ve{x}}'\in\Omega(\Box)}l^3 \nabla_{\ve{x}'}G(\ve{x},\ve{x}')\cdot\ve{p}(\hat{\ve{x}}')+\sum_{\hat{\ve{x}}'\in\Omega(\lr)}l^2G(\ve{x},\hat{\ve{x}}')\sigma(\hat{\ve{x}}'),
    \end{split}
\end{equation}
where to get the penultimate line, we used the charge neutrality of the unit cells.

Now for $l$ small enough, we can replace the sum with an integration. This gives us,
\begin{equation}
    \begin{split}
        \Phi_l(\ve{x})&=\int_{\Omega} \nabla_{\ve{x}'}G(\ve{x},\ve{x}')\cdot\ve{p}({\ve{x}}') \dm\ve{x}'+\int_{\partial\Omega}G(\ve{x},\ve{x}')\sigma(\ve{x}') \dm S_{{\ve{x}'}}\\
                 &=\int_{\Omega} G(\ve{x},\ve{x}')\big(-\divergence_{\ve{x}'}\ve{p}+\mathbbm{1}_{\partial\Omega}(\ve{p}\cdot \ve{n}+\sigma)\big) \dm\ve{x}'.
    \end{split}
\end{equation}
Comparing the above with \eqref{formal-Int} gives us the corresponding result for \eqref{1.3}.

\section{Notation}\label{Notation}
In this section, we introduce the notations that will be utilized throughout the paper. We begin by providing key definitions,
\begin{enumerate}
\item[] $\Scr{Z}$ : set of integers.
\item[] $\Scr{R}$ : set of reals.
\item[] $\{\}$ : the null set.
\item[] $B(R,f)$: Denotes a ball of radius $R$ centered at $f$.
\item[] $l\Box$ : a representative scaled unit cell of dimensions $l\mathsf{x}l\mathsf{x}l$.
\item[] $\Box\equiv l\Box\big\vert_{l=1} :$ represents a unit cell.
\item[] $\Omega(\Box)$ : denotes the set of all cells inside $\Omega$.
\item[] $l\lr$ : a representative scaled partial unit cell. A scaled partial unit cell is defined as a scaled unit cell not entirely contained in the domain.
\item[] $\lr\equiv l\lr\big\vert_{l=1} :$ this represents a partial unit cell. A partial unit cell is defined as a unit cell not entirely contained in the domain.
\item[] $\Omega(\lr)$: denotes the set of all partial unit cells in $\Omega$.
\item[] $\scr{I}$ : this is the indexing set which contains the sequence of $l$'s going to 0.
\item[] $|\cdot|,~||\cdot||$ : has been used in multiple ways to denote different concepts with the context clarifying the meaning.
\begin{enumerate}
    \item to denote distance.
    \item cardinality of a set.
    \item Lebesgue measure of a set.
\end{enumerate}
\item[] $\scr{L}$ : the lattice constructed from lattice vectors $\{\ve{e}_i\}_1^3$. Vectors are denoted using bold letters.
\item[] $l\scr{L}$ : the lattice constructed from lattice vectors $l\ve{e}_i$.
\item[] $\text{PDE}$ : is shorthand for partial differential equation.
\item[] $L^p$ : space of measurable functions, the integral of whose p$^{th}$ power of its absolute value is finite.
\begin{enumerate}
    \item[] For $p=\infty$, the essential supremum is required to be bounded.
\end{enumerate}
\item[] $\scr{C}^k$ : space of functions with continuous first $k$ derivatives.
\item[] $W^{k,p}$ : space of functions denoting functions and their first $k$ derivatives being in $L^p$. 
\begin{enumerate}
    \item[] For $p=2$, the space is represented by the short hand $H^k$.
\end{enumerate}
\item[] $D$ : space of distributions.
\item[] $\scr{C}_{\#}^{\infty}$ : space of periodic functions which are infinitely differentiable.
\item[] $f\vert_U$ : This is the value of the function restricted to the set $U$.
\item[] $f(x,\cdot)$ : shorthand for a function $f(x,y)$, where the first argument ($x$) is fixed, and the second argument is allowed to vary.
\item[] $\rho_l(\ve{x})$ : sequence of charge densities.
\item[] ${\rho}_l(\hat{\ve{x}}_l,\ve{z})$ : sequence of two scale identified charge distribution. The two scales are identified with the corner map.
\item[] $\tilde{\rho}_l(\hat{\ve{x}}_l,\ve{y})$ : sequence of dipole scaled and two scale identified charge distribution. The two scales are identified with the corner map. The fast variable is appropriately scaled.
\item[] $\sigma_l$ : surface charge distribution due to partial unit cells $\forall~l\in\scr{I}$. Defined in \eqref{(5.17)} and \eqref{weaksurf}.
\item[] $\sigma$ : limit of surface charge distribution due to partial unit cells. Limit of \eqref{(5.17)} and \eqref{weaksurf}, first occurring in \eqref{2ScaleThm}.
\end{enumerate}
\vspace{1em}

We define the Minkowski sum between two sets $P$ and $Q$ as 
\begin{equation}
P\oplus Q:=\{p+q\vert p\in P,~q\in Q\}
\end{equation}
where for our purposes, we require that $P,Q\subset\Scr{R}^d$ and we use $x$ to denote the singleton set $\{x\}$, provided no confusion arises.

We define a lattice $\scr{L}$ for this problem,
\begin{equation}
\scr{L}:=\big\{\ve{x}\in\Scr{R}^3\big\vert \ve{x}=\sum_{i}x^i\ve{e}_i,~x^i\in\Scr{Z},~i=1,2,3\big\}
\end{equation}
where $\{\ve{e}_i\}_1^3$ are linearly independent lattice vectors describing the lattice. The corresponding dual lattice ($\scr{RL}$) and dual vectors $\{\ve{e}^j\}_1^3$ are defined by the orthogonality relation $\ve{e}_i\cdot\ve{e}^j=\delta^j_{i}$. We scale $\scr{L}$ by a factor $l\in(0,1]$ and obtain a lattice $l\scr{L}$ with corresponding lattice vectors $\{\ve{e}_{l,i}\}_1^3$ defined as
\beol{\ve{e}_{l,i}=l\ve{e}_i,~~i=1,2,3.}

To define a unit cell, it is also necessary to specify a corresponding corner map. We select ${\Box}^{\prime} \subset \mathbb{R}^3$ bounded and $\ve{f} \in \mathbb{R}^3$ such that the following conditions are satisfied:
\begin{enumerate}
    \item[i.] The unit cells covers $\Scr{R}^3$,
    \beolnn{\Scr{R}^3=\bigcup_{\ve{i}\in l\scr{L}}\ve{i}\oplus(l\ve{f}\oplus l{\Box}^{\prime}),}
    \item[ii.] The unit cells have no overlap
    \beolnn{\forall\ve{i}\neq\ve{j}\in l\scr{L},~~~~\ve{i}\oplus(l\ve{f}\oplus l{\Box}^{\prime})\cap\ve{j}\oplus(l\ve{f}\oplus l{\Box}^{\prime})=\{\}}
    \item[iii.] The unit cells are bounded, ie, $\exists~R>0$ such that 
    \beolnn{\ve{f}\oplus\Box'\subset B(R,\ve{f}),}
    where $B(R,\ve{f})$ is a ball of radius $R$ centered at $\ve{f}$.
\end{enumerate}
Observe that the first and second conditions are standard requirements for defining a unit cell. The third condition is necessary to ensure that different choices of unit cells yield a consistent notion of packing dimension or box-counting dimension, also referred to as the entropy index, which is intrinsically related to the Hausdorff dimension.

\rmrk{The boundedness of the unit cell helps prevent choices of unit cells with disjoint subsets that are infinitely far apart. Moreover, by (i) and (ii), we know that the unit ball contains a countable number of these unit cells and, equivalently, must be contained by a countable number of unit cells. If the unit cells were not measurable, then the ball would not be measurable either. This property is essential for defining measurable functions and computing the polarization on the unit cell.}

The scaled unit cell is defined as $l\Box:=l\ve{f}\oplus l{\Box}^{\prime}$ and the corner map is defined as
\beol{\hat{\ve{x}}_l:=\ve{i}+l\ve{f}~\text{if}~\ve{x}\in\Omega~\&~\exists\ve{i}\in l\scr{L}~\text{such that}~\ve{x}\in\ve{i}\oplus l\Box,}
where the subscript $l$ in the corner map is used to denote the scaling of the lattice. If no confusion arises, we will drop the subscript. We define the simple unit cell as $\Box:=[0,1)^3$ with $\ve{f}=0$.

A scaled partial cell of $l\scr{L}$, denoted $l\lr$, is defined as a scaled unit cell that is not entirely contained in the domain
\begin{equation}
{l\lr:=\big\{\ve{z}\in l\Box\big\vert\exists \ve{i}\in l\scr{L}~\text{such that }\ve{i}+ \ve{z}\in\Omega~\text{but}~\ve{i}\oplus l\Box\cap\Omega^c\neq\{\}\big\}.}
\end{equation}
The definition of a unit-cell and partial unit cell can be obtained from their scaled counterparts by re-scaling each by $\dfrac{1}{l}$.

The collection of scaled unit cells completely contained inside $\Omega$ is referred to as $\Omega(\Box)$. The remaining partial unit cells are referred to as $\Omega(\lr):=\Omega\backslash\Omega(\Box)$. Note that for some $\hat{\ve{x}}_l\oplus l\Box\in\Omega(\Box)\cup\Omega(\lr)$, it may turn out that $\hat{\ve{x}}_l\not\in\Omega$. Thus, we write $\hat{\ve{x}}_l\in\Omega(\Box)$ and $\hat{\ve{x}}_l\in\Omega(\lr)$ as shorthand for $\hat{\ve{x}}_l\oplus l\Box\subset\Omega(\Box)$ and $\hat{\ve{x}}_l\oplus l\lr\subset\Omega(\lr)$, respectively.

We define the family of charge densities as $\{\rho_l\}_{l \in \scr{I}}$ for each lattice $l\scr{L}$, where $\scr{I}$ represents an indexing set for the different values of the scaling parameter. To establish the polarization theorem \eqref{2ScaleThm}, we require that the charge density is bounded and has zero mean on each subset $\hat{\ve{x}}_l \oplus l\Box \subset \Omega$. However, the charge density must satisfy a stricter condition of local periodicity in order to derive the local Poisson equation \eqref{Poisson} associated with the corrector field; this assumption is a direct consequence of the applicability of Allaire's two-scale theorem \cite{allaire1992homogenization}. It is worth noting that while the polarization theorem \eqref{2ScaleThm} remains valid even when the charge density is not locally periodic, the resulting quantity will not correspond to a macro-scale quantity.

As noted earlier, we require charge neutrality to obtain a convergent energy density. We impose charge neutrality by requiring that 
\begin{equation}
\label{eqn:charge-neutrality}
    {\int_{\hat{\ve{x}}_l\oplus l\Box}\rho_l(\ve{x})\dm\ve{x}=0\qquad\qquad \forall \hat{\ve{x}}_l\oplus l\Box\subset\Omega,~\forall l\in\scr{I}.}
\end{equation}
We employ the dipole scaling that requires that as $l\to 0$, we proportionately increase the magnitude of $\rho_l(\ve{x})$. We introduce the scaled charge density $\rho_l(\ve{x}):=\frac{1}{l}\tilde{\rho}_l(\ve{x})$, which we refer to as the dipole scaling. To justify this scaling, we utilize the corner map to express the charge distribution as --- $\rho_l(\ve{x})=\rho_l(\hat{\ve{x}}_l+\ve{z})=:\rho_l(\hat{\ve{x}}_l,\ve{z})$ where $\ve{z}\in l\Box$. We now introduce the dipole scaling in the second argument as
\begin{equation}
{\rho_l(\ve{x})=\rho_l(\hat{\ve{x}},\ve{z})=\dfrac{1}{l}\tilde{\rho}_l(\hat{\ve{x}},\dfrac{\ve{z}}{l}),\label{2scaleiden}}
\end{equation}
where in the above expression, we are locally scaling the magnitude of the charge density inversely with the size of the unit cell.

    Following \cite{sen2024nonuniqueness}, we provide a simple example.
    \begin{exmp}
    Consider 
    \begin{equation}
        \rho_l(x)=\dfrac{\sin(2\pi x/l)}{l}
    \end{equation}
    and choose the domain as $[0,L]$; the scaled unit cell is $[0,1]$.
    The corner map is 
    $$x=\hat{x}+ly=nl+z=nl+ly,$$
    where $n=[x/l]$, and $[~]$ denotes the floor function.
    
    Using the two-scale identification, we have:
    \begin{equation}
    \tilde{\rho}_l(x,y)=\sin(2\pi n+2\pi y)=\sin(2\pi y)
    \end{equation}

    To check condition \eqref{3.3}, we compute:
    \begin{equation}
        \int_{[0,1]}\tilde{\rho}_l^2 \dm y
        =
        \int_{[0,1]}\sin^2(2\pi y) \dm y 
        = \frac{1}{2}.
    \end{equation}
    
    To check the charge neutrality condition \eqref{eqn:charge-neutrality}, we compute:
    \begin{equation}
        \int_{nl+[0,l]}\rho_l \dm x
        =
        \int_{nl+[0,l]}\dfrac{\sin(2\pi x/l)}{l} \dm x
        =
        \int_{[0,l]}\dfrac{1}{l} \sin(2\pi z/l) \dm z
        =
        0,
    \end{equation}
    where we set $x=nl+z$ and use that it the integrand is $l$-periodic and being integrated over its period.
    \end{exmp}
Further examples are available in \cite{sen2024nonuniqueness}.


\section{Problem Formulation}\label{PF}

We examine two problems; the first is posed on $\Scr{R}^3$ arising from a compactly supported charge distribution, while the second is an interior problem formulated to determine the electrostatic potential within a bounded domain $\Omega$. We determine the homogenized and corrector equations corresponding to the Poisson problem for each.

The boundary is decomposed as $\partial\Omega=\Gamma_d\cup\Gamma_n$ with $\Gamma_d\cap\Gamma_n=\{\}$, where $\partial\Omega$ is the boundary, and $\Gamma_d$ and $\Gamma_n$ are parts of the boundary where Dirichlet and Neumann boundary conditions are specified respectively; $\Gamma_n$ is assumed $\scr{C}^1$.

In the first problem, we consider a sequence of charge distributions $\{\tilde{\rho}_l\}_{l\in\scr{I}}$ in $\Scr{R}^3$ such that $\text{Sppt}(\tilde{\rho}_l)=\Omega~\forall l\in\scr{I}$. The charge distribution is locally periodic with respect to the lattice $l\scr{L}$ with zero mean and $L^2$ bound in each cell. The potential is specified on $\Gamma_d$ with decay boundary conditions at infinity. The family of PDEs are stated as follows:
\begin{equation}
\begin{split}
-\Delta &\Phi_l(\ve{x})=\dfrac{\tilde{\rho}_l(\ve{x})}{l}~~\text{in}~\Scr{R}^3\backslash\Gamma_d,\\
	 &\Phi_l=f~~~\text{on}~\Gamma_d,\\
	 &\Phi_l=0~~~\text{as}~~|\ve{x}|\to\infty.\label{Prob0}\\
     &\Phi_l\in L^2_{loc}(\Scr{R}^3)~~\text{and}~~\nabla\Phi_l\in L^2(\Scr{R}^3),~\forall l\in\scr{I}
\end{split},
\end{equation}
where the last condition ensures that the energy is finite.

In the second problem, we analyze a sequence of charge distributions $\{\tilde{\rho}_l\}_{l\in\scr{I}}$ defined on a domain $\Omega$, which is locally periodic with respect to the lattice $l\scr{L}$ and is $L^2$ bounded with zero mean on each unit cell. A family of partial differential equations associated with these charge distributions $\{\tilde{\rho}_l\}$ is studied, complemented with appropriate boundary conditions. The family of PDEs are stated as follows
\begin{equation}
\begin{split}
-\Delta &\Phi_l(\ve{x})=\dfrac{\tilde{\rho}_l(\ve{x})}{l}~~\text{in}~\Omega,\\
	 &\Phi_l=f~~~\text{on}~\Gamma_d,\\
	 \dfrac{\partial}{\partial n}&\Phi_l=g~~~\text{on}~\Gamma_n,\\
     \nabla&\Phi_l~\text{is bounded uniformly in $L^2(\Omega),~\forall l\in\scr{I}$},\label{Prob}
\end{split}
\end{equation}
where we are imposing the same boundary conditions for each of the problems. The last condition corresponds to the energy being finite.

Since our objective is to compute the polarization in each unit cell, it is reasonable to require that the charge density be $L^2$ bounded within each unit cell. Furthermore, as the charge density is defined over the entire domain, this local $L^2$ boundedness is expected to imply a global $L^2$ bound on the entire domain. To establish this relationship, we impose the condition that the charge density is $L^2$ bounded within each unit cell.
\begin{equation}
    {||\tilde{\rho}_l||_{L^2(\hat{\ve{x}}_l\oplus l\Box)}:=\bigg[\int_{\Box}\tilde{\rho}^2_l(\hat{\ve{x}}+l\ve{y})\dm\ve{y}\bigg]^{\frac{1}{2}}\leq C,~~\forall \ve{x}\in\Omega ,~\forall l\in\scr{I}}.\label{3.3}
\end{equation}
Note that $C$ is independent of $l$.

\begin{theorem}\label{GlobalBound}
    If $||\tilde{\rho}_l||_{L^2(\hat{\ve{x}}_l\oplus l\Box)}<C$ for all $\ve{x}\in\Omega$ and for all $l\in\scr{I}$, then $||\tilde{\rho}_l||_{L^2(\Omega)}<C$ for all $l\in\scr{I}$ and $C$ independent of $l$.
\end{theorem}

\begin{proof}
We demonstrate that the above bounds also implies that $\exists~C>0$ such that $||\tilde{\rho}_l||_{L^2(\Omega)}<C$
\begin{equation}
\begin{split}
\int_{\Omega}\tilde{\rho}^2_l\dm\ve{x}&=\sum_{\hat{\ve{x}}\in\Omega(\Box)}l^3\int_{\Box}\tilde{\rho}^2_l(\hat{\ve{x}},\ve{y})\dm\ve{y}+\sum_{\hat{\ve{x}}\in\Omega(\lr)}l^3\underbrace{\int_{\lr}\tilde{\rho}^2_l(\hat{\ve{x}},\ve{y})\dm\ve{y}}_{\leq\int_{\Box}\tilde{\rho}^2_l(\hat{\ve{x}},\ve{y})\dm\ve{y}}\\
&\leq \underbrace{\sum_{\hat{\ve{x}}\in\Omega(\Box)\cup\Omega(\lr)}}_{\#(\Omega/l^3)} l^3\underbrace{\int_{\Box}\tilde{\rho}^2_l(\hat{\ve{x}},\ve{y})\dm\ve{y}}_{\leq C} \leq C|\Omega|+o(1)\label{forward}
\end{split}
\end{equation}
where $\#(\Omega/l^3)$ represents the number of boxes of size $l$ covering $\Omega$. The fact that $\sum_{\#(\Omega/l^3)}l^3$ is an approximation to $|\Omega|$ form `above' is captured by the $o(1)$ term, i.e., we have
\beol{\sum_{\hat{\ve{x}}\in\Omega(\Box)}l^3\leq |\Omega|\leq \sum_{\hat{\ve{x}}\in\Omega(\Box)\cup\Omega(\lr)}l^3~~\forall l\in\scr{I};}
we justify that we will get $|\Omega|$ in the limit by adding an $o(1)$ term which goes to 0 as $l\to 0$.     
\end{proof}

Recall the two-scale identification process outlined in \eqref{2scaleiden}. Since $\{\rho_l\}_{\scr{I}}$ is defined on $\Omega$, it depends on a single variable. In contrast, a two-scale limit requires two variables. The mapping of a single-variable function to a two-variable function is achieved through the two-scale identification process. To extract a well-defined two-scale limit, it is necessary to demonstrate that for a given sequence in $L^2(\Omega)$, there exists a corresponding sequence that is bounded in the space $L^2(\Omega \times \Box)$. This allows extraction of a weak limit in $L^2(\Omega \times \Box)$. 

\begin{theorem}\label{2sclidenbdd}
  Let $\{\tilde{\rho}_l\}_{l \in \scr{I}}$ be a bounded sequence in $L^2(\Omega)$ be a given sequence of functions. For functions defined through the two-scale identification process \eqref{2scaleiden}, it follows that $\{\tilde{\rho}_l(\hat{\ve{x}}, \ve{y})\}_{l \in \scr{I}}$ is a bounded sequence in $L^2(\Omega \times \Box)$.
\end{theorem}

\begin{proof}
Noting that $||\tilde{\rho}_l||_{L^2(\Omega)} \leq C$,
\begin{equation}
\begin{split}
C\geq\int_{\Omega}\tilde{\rho}^2_l\dm\ve{x}&=\sum_{\hat{\ve{x}}\in\Omega(\Box)}l^3\int_{\Box}\tilde{\rho}^2_l(\hat{\ve{x}},\ve{y})\dm\ve{y}+\sum_{\hat{\ve{x}}\in\Omega(\lr)}l^3\int_{\lr}\tilde{\rho}^2_l(\hat{\ve{x}},\ve{y})\dm\ve{y}\\
&= \sum_{\hat{\ve{x}}\in\Omega(\Box)\cup\Omega(\lr)} l^3\int_{\Box}\mathbbm{1}_{\Omega}(\ve{x})\tilde{\rho}^2_l(\hat{\ve{x}},\ve{y})\dm\ve{y},
\end{split}
\end{equation}
where the characteristic function $\mathbbm{1}_{\Omega}$ is employed to substitute the integral over the partial unit cells with the integrals over the corresponding full unit cells.

Observe that the identified sequence $\tilde{\rho}_l(\hat{\ve{x}}, \ve{y})$ remains constant with respect to the first argument. Consequently, the summation can be replaced by an integral, yielding:
\begin{equation}
C\geq\int_{\Omega\times\Box}\tilde{\rho}^2_l(\hat{\ve{x}},\ve{y})\dm\ve{x}\dm\ve{y},
\end{equation}
where the characteristic function on $\Omega$ is omitted, as the integration is already performed over $\Omega$. Consequently, we have demonstrated that the identified sequence is bounded in $L^2(\Omega \times \Box)$.    
\end{proof}

Our formulation is presented with respect to the scaled charge density $\tilde{\rho}_l(\ve{x}) \left(= l\rho_l(\ve{x})\right)$. This scaling is chosen to ensure the extraction of a dipole moment in the limit. If the objective were to extract a quadrupole moment in the homogenized limit, the appropriate scaling would involve scaling the charge density with $l^{-2}$ while imposing the conditions that the zeroth- and first-order moments vanish. A similar approach can be applied to derive all higher-order moments.

The homogenization process is organized as follows: for each $l \in \scr{I}$, the lattice $l\scr{L}$ is introduced, and the domain $\Omega$ is partitioned into two distinct subsets: a collection of unit cells fully contained within $\Omega$, denoted as $\Omega(\Box)$, and a collection of partial unit cells that are not entirely contained within $\Omega$, denoted as $\Omega(\lr)$. Leveraging charge neutrality within each unit cell, the charge distribution in each unit cell is replaced by the net polarization. In contrast, the partial unit cells, which are not charge neutral, contribute a surface charge denoted by $\sigma$. This surface charge is assigned to the boundary $\partial\Omega$, resulting in a homogenized problem defined on $\Omega(\Box) \cup \partial\Omega$. In the limit as $l \to 0$, the homogenized problem is posed on the entire domain $\Omega$. For a visual explanation of this process, along with illustrative examples, the interested reader is directed to Section 2 of \cite{sen2024nonuniqueness}.

\section{Weak Formulation}\label{WF}

We begin by presenting the weak formulation for problem \eqref{Prob0}. To this end, we multiply \eqref{Prob0} by a test function defined as $[\psi]_l(\ve{x}) := \psi_0(\ve{x}) + l\psi_1\left(\ve{x}, \frac{\ve{x}}{l}\right)$. Integrating over $\Scr{R}^3$, we obtain:
\begin{equation}
{\int_{\Scr{R}^3}-\Delta\Phi_l(\ve{x})[\psi]_l(\ve{x})\dm\ve{x}=\int_{\Omega}\dfrac{\tilde{\rho}_l(\ve{x})}{l}[\psi]_l(\ve{x})\dm\ve{x},}
\end{equation}
where $\psi_1\in D\big(\Omega,\scr{C}^{\infty}_{\#}(\Box)\big)$ and $\psi_0\in V:=\{u\in D(\Scr{R}^3)~\vert~u(\Gamma_d)=0\}$. 
    Additionally, we restrict our attention to domains for which the surface charge $\sigma$ is in $L^1(\partial\Omega)$; this is to enable us to go from \eqref{4.3o} to \eqref{WeakForm0} using the result \eqref{2ScaleThm} below.
Integrating by parts,
\begin{equation}            
\int_{\Scr{R}^3}\nabla\Phi_l\cdot\big(\nabla_{\ve{x}}\psi_0+\nabla_{\ve{y}}\psi_1+l\nabla_{\ve{x}}\psi_1\big)\dm\ve{x}=\int_{\Omega}\psi_0\dfrac{\tilde{\rho}_l}{l}\dm\ve{x}+\int_{\Omega}\psi_1\tilde{\rho}_l\dm\ve{x}.
\end{equation}

We pass $l$ to the limit and extract the following equation:
\begin{equation}
    \begin{split}
&\int_{\Scr{R}^3\times\Box}\big(\nabla_{\ve{x}}\Phi_0+\nabla_{\ve{y}}\Phi_1\big)\cdot\big(\nabla_{\ve{x}}\psi_0+\nabla_{\ve{y}}\psi_1\big)\dm\ve{x}\dm\ve{y}=\lim_{l\to0}\bigg[\int_{\Omega}\psi_0\dfrac{\tilde{\rho}_l}{l}\dm\ve{x}+\int_{\Omega}\psi_1\tilde{\rho}_l\dm\ve{x}\bigg],        
    \end{split}\label{4.3o}
\end{equation}
where $\Phi_1(\ve{x},\ve{y})\in~L^2[\Scr{R}^3\backslash\Gamma_d,H^1_{\#}(\Box)/\Scr{R}]$ \cite{allaire1992homogenization}. We refer the reader to \cite{tian2012dielectric} for the reasoning used to demonstrate the existence of the 2-scale limit and $\Phi_1(\ve{x},\ve{y})$.

By utilizing the weak limit of the first term on the right-hand side, the proof of which is provided in Section \ref{2sTm}, along with the two-scale limit of the final term, we derive the following:
\begin{equation}
    \begin{split}
&\int_{\Scr{R}^3\times\Box}\big(\nabla_{\ve{x}}\Phi_0(\ve{x})+\nabla_{\ve{y}}\Phi_1(\ve{x},\ve{y})\big)\cdot\big(\nabla_{\ve{x}}\psi_0(\ve{x})+\nabla_{\ve{y}}\psi_1(\ve{x},\ve{y})\big)\dm\ve{x}\dm\ve{y}\\
&=\int_{\partial\Omega}\psi_0(\ve{x})\sigma(\ve{x}) \dm S_{\ve{x}}+\int_{\Omega}\nabla_{\ve{x}}\psi_0(\ve{x})\cdot\ve{p}_0(\ve{x}) \dm\ve{x}+
\int_{\Omega\times\Box}\psi_1(\ve{x},\ve{y})\tilde{\rho}_0(\ve{x},\ve{y})\dm\ve{y}\dm\ve{x},   
    \end{split}\label{WeakForm0}
\end{equation}
where in the above, we utilized a compactness theorem to extract a two-scale function $\tilde{\rho}_0(\ve{x},\ve{y})$ \cite{allaire1992homogenization}, the polarization $\ve{p}_0(\ve{x})$ and a surface charge $\sigma(\ve{x})$. The polarization $\ve{p}_0$ is defined as $\ve{p}_0(\ve{x}):=\int_{\Box}\ve{y}\tilde{\rho}_0(\ve{x},\ve{y})\dm\ve{y},$ while the surface charge $\sigma$ is defined as $\sigma(\ve{x})=\lim_{l\to0}\frac{1}{K_l}\int_{\lr}\tilde{\rho}_l(\ve{x},\ve{y})\dm\ve{y}$ where $K_l$ has been defined in Section \ref{SurfCrgDnsty}.

Setting $\psi_0(\ve{x})=0$ yields:
\begin{equation}
    \int_{\Scr{R}^3\times\Box}[\nabla_{\ve{x}}\Phi_0(\ve{x})+\nabla_{\ve{y}}\Phi_1(\ve{x},\ve{y})]\cdot\nabla_{\ve{y}}\psi_1(\ve{x},\ve{y})\dm\ve{y}\dm\ve{x}=\int_{\Omega\times\Box}\psi_1(\ve{x},\ve{y})\tilde{\rho}_0(\ve{x},\ve{y})\dm\ve{y}\dm\ve{x},
\end{equation}
performing integration by parts,
\begin{equation}
    \begin{split}
&\int_{\Scr{R}^3\times\Box}\nabla_{\ve{y}}\cdot\big(\psi_1[\nabla_{\ve{x}}\Phi_0+\nabla_{\ve{y}}\Phi_1]\big)\dm\ve{y}\dm\ve{x}-\int_{\Scr{R}^3\times\Box}\Delta_{\ve{y}}\Phi_1\psi_1\dm\ve{y}\dm\ve{x}=\int_{\Omega\times\Box}\psi_1\tilde{\rho}_0\dm\ve{y}\dm\ve{x},
    \end{split}
\end{equation}
where, in the above, we have utilized the fact that $\nabla_{\ve{y}}\cdot\nabla_{\ve{x}}\Phi_0(\ve{x})=0$.  By leveraging the property that the integral of the divergence of a periodic function over its period is zero, we derive the following:
\begin{equation}
\int_{\Scr{R}^3\times\Box}\Big[\psi_1\Delta_{\ve{y}}\Phi_1+\mathbbm{1}_{\Omega}\psi_1\tilde{\rho}_0\Big]\dm\ve{y}\dm\ve{x}=0~~~\forall\psi_1(\ve{x},\ve{y})\in L^2[\Scr{R}^3\backslash\Gamma_d,H^1_{\#}(\Box)/\Scr{R}].
\end{equation}

Utilizing the density of $\psi_1$, we derive the following PDE:
\begin{equation}
    \begin{split}
        &-\Delta_{\ve{y}}\Phi_1(\ve{x},\ve{y})=\mathbbm{1}_{\Omega}(\ve{x})\tilde{\rho}_0(\ve{x},\ve{y})~\text{in}~\Scr{R}^3\backslash\Gamma_d\times\Box\\
&~~~\Phi_1(\ve{x},\ve{y}) \in~~L^2[\Scr{R}^3\backslash\Gamma_d;H^1_{\#}(\Box)/\Scr{R}],
    \end{split}
    \label{Poisson0}
\end{equation}
where it is noted that $\mathbbm{1}_{\Omega}(\partial\Omega)=0$, highlighting that the boundary does not influence the cell equation. As will become evident later, the discontinuity in the charge density at $\partial\Omega$ impacts only the macro-scale equation.

Setting $\psi_1(\ve{x},\ve{y})=0$ in \eqref{WeakForm0},
\bml{&\int_{\Scr{R}^3\times\Box}\nabla_{\ve{x}}\psi_0\cdot(\nabla_{\ve{x}}\Phi_0+\nabla_{\ve{y}}\Phi_1)\dm\ve{y}\dm\ve{x}=\int_{\partial\Omega}\psi_0\sigma \dm S_{\ve{x}}
+\int_{\Omega}\nabla_{\ve{x}}\psi_0\cdot \ve{p}_0\dm\ve{x},}
which, after simplification, yields:
\beol{\int_{\Scr{R}^3}\nabla_{\ve{x}}\Phi_0\cdot\nabla_{\ve{x}}\psi_0\dm\ve{x}=\int_{\partial\Omega}\psi_0\big(\sigma+\ve{p}_0\cdot \ve{n}\big) \dm S_{\ve{x}}-\int_{\Omega}\psi_0\divergence\ve{p}_0\dm\ve{x},\label{(4.10)}}
where we have employed the fact that $\nabla_{\ve{y}}\Phi_1$ had $0$-mean in $\Box$.

To derive the strong form, we set $\text{Sppt}(\psi_0)\subset\Scr{R}^3\backslash\overline{\Omega}$. Substituting into \eqref{(4.10)} yields:
\beol{\int_{\Scr{R}^3}\Delta_{\ve{x}}\Phi_0(\ve{x})\psi_0(\ve{x})\dm\ve{x}=0,}
using the density of $\psi_0$, 
\beol{\Delta_{\ve{x}}\Phi_0(\ve{x})=0~\text{in}~\Scr{R}^3\backslash\overline{\Omega}.}

We now set $\text{Sppt}(\psi_0)\subset\Omega$.  Substituting into \eqref{(4.10)} results in:
\beol{\int_{\Omega}\Delta_{\ve{x}}\Phi_0\psi_0\dm\ve{x}=\int_{\Omega}\psi_0\divergence\ve{p}_0\dm\ve{x},}
using the density of $\psi_0$, we get
\beol{\Delta_{\ve{x}}\Phi_0(\ve{x})=\divergence\ve{p}_0(\ve{x})~~\text{in}~\Omega.}

Next, we define $\hat{\Omega}:=\text{Sppt}(\psi_0)\supset\partial\Omega$. Utilizing the condition that $\psi_0\in V$, we restrict the choice to include only  $\partial\Omega\backslash\Gamma_d$. Consequently, equation \eqref{(4.10)} is transformed into:

\begin{equation}
    \begin{split}
&\int_{\hat{\Omega}}\nabla_{\ve{x}}\Phi_0\cdot\nabla_{\ve{x}}\psi_0\dm\ve{x}=\int_{\partial\Omega\backslash\Gamma_d}\psi_0\big(\sigma+\ve{p}_0\cdot\ve{n}\big)\dm S_{\ve{x}}-\int_{\Omega\cap\hat{\Omega}}\psi_0\divergence\ve{p}_0\dm\ve{x},\\
\implies&\int_{\big(\hat{\Omega}\cap\Omega^c\big)\cup\big(\hat{\Omega}\cap\Omega\big)}\nabla_{\ve{x}}\Phi_0\cdot\nabla_{\ve{x}}\psi_0\dm\ve{x}=\int_{\partial\Omega\backslash\Gamma_d}\psi_0\big(\sigma+\ve{p}_0\cdot\ve{n}\big)\dm S_{\ve{x}}-\int_{\Omega\cap\hat{\Omega}}\psi_0\divergence\ve{p}_0\dm\ve{x},\\
\implies&\int_{\partial(\hat{\Omega}\cap\Omega^c)}\dfrac{\partial\Phi_0}{\partial n}\psi_0\dm S_{\ve{x}}-\int_{\hat{\Omega}\cap\Omega^c}\Delta_{\ve{x}}\Phi_0\psi_0\dm\ve{x}+\int_{\partial(\hat{\Omega}\cap\Omega)}\dfrac{\partial\Phi_0}{\partial n}\psi_0\dm S_{\ve{x}}-\int_{\hat{\Omega}\cap\Omega}\Delta_{\ve{x}}\Phi_0\psi_0\dm\ve{x}\\
&=\int_{\partial\Omega\backslash\Gamma_d}\psi_0\big(\sigma+\ve{p}_0\cdot\ve{n}\big)\dm S_{\ve{x}}-\int_{\Omega\cap\hat{\Omega}}\psi_0\divergence\ve{p}_0\dm\ve{x}.
    \end{split}
\end{equation}

Utilizing the fact that $\Delta_{\ve{x}} \Phi_0 = 0$ in $\Omega^c$, while $\Delta_{\ve{x}} \Phi_0 = \divergence \, \ve{p}_0$ in $\Omega$, and noting that $\psi_0$ is non-zero on $\partial\Omega \setminus \Gamma_d$, we obtain:
\beol{\int_{\partial\Omega\backslash\Gamma_d}\bigg(\Bigg\llbracket\dfrac{\partial\Phi_0}{\partial n}\Bigg\rrbracket-\sigma-\ve{p}_0\cdot\ve{n}\bigg)\psi_0=0.}

Using density of $\psi_0$,
\beol{\Bigg\llbracket\dfrac{\partial\Phi_0}{\partial n}\Bigg\rrbracket=\sigma+\ve{p}_0\cdot\ve{n}=\sigma_t~~\text{on}~\partial\Omega\backslash\Gamma_d,}
where $\sigma_t$ will be referred to as the total surface charge.

Combining all the equations, we obtain:
\begin{equation}
    \begin{split}
\Delta_{\ve{x}}&\Phi_0(\ve{x})=\mathbbm{1}_{\Omega}\divergence\ve{p}_0~~\text{in}~\Scr{R}^3,\\
	 &\Phi_0=f~\text{on}~\Gamma_d,\\
	 &\Bigg\llbracket\dfrac{\partial\Phi_0}{\partial n}\Bigg\rrbracket=\ve{p}_0\cdot \ve{n}+\sigma~\text{on}~\partial\Omega\backslash\Gamma_d.        
    \end{split}\label{Soln0}
\end{equation}
	 
\begin{rmrk}
{Observe that the surface integral term in the polarization theorem provides a surface charge for the entire boundary. However, in the strong form, the surface charge term appears only on $\partial\Omega \setminus \Gamma_d$. This raises the question of where the remaining surface charge is accounted for. The answer lies in the interpretation of $\Gamma_d$ as a model for the electrodes of a battery. The battery enforces a potential $f$ on $\Gamma_d$ by supplying charge. Specifically, it provides a charge $\sigma_f$ such that $\sigma_f + \ve{p}_0 \cdot \ve{n} + \sigma$ corresponds to the charge determined by the application of the Dirichlet-to-Neumann map, thereby compensating for the surface charge distribution arising from the partial unit cells.
}
\end{rmrk}

We now present the weak formulation for problem \eqref{Prob}. The governing equation is multiplied by a test function $[\psi]_l(\ve{x})$. Integrating over $\Omega$, we obtain:
\begin{equation}
{\int_{\Omega}-\Delta\Phi_l(\ve{x})[\psi]_l(\ve{x})\dm\ve{x}=\int_{\Omega}\dfrac{\tilde{\rho}_l(\ve{x})}{l}[\psi]_l(\ve{x})\dm\ve{x},}
\end{equation}
where $\psi_1 \in D\big(\Omega, \scr{C}^{\infty}_{\#}(\Box)\big)$ and $\psi_0 \in V := \{u \in D(\Omega) \mid u\left({\Gamma_d}\right) = 0\}$. By applying integration by parts, we obtain:
\begin{equation}
    \begin{split}        &\int_{\Omega}\nabla\Phi_l(\ve{x})\cdot\big(\nabla_{\ve{x}}\psi_0(\ve{x})+\nabla_{\ve{y}}\psi_1(\ve{x},\dfrac{\ve{x}}{l})+l\nabla_{\ve{x}}\psi_1(\ve{x},\dfrac{\ve{x}}{l})\big)\dm\ve{x}\\
&=\int_{\partial\Omega}\dfrac{\partial\Phi_l}{\partial n}\big(\psi_0(\ve{x})+l\psi_1(\ve{x},\dfrac{\ve{x}}{l})\big)\dm S_{\ve{x}}+\int_{\Omega}\psi_0(\ve{x})\dfrac{\tilde{\rho}_l(\ve{x})}{l}\dm\ve{x}+\int_{\Omega}\psi_1(\ve{x},\dfrac{\ve{x}}{l})\tilde{\rho}_l(\ve{x})\dm\ve{x},\\
\implies&\int_{\Omega}\nabla\Phi_l\cdot\big(\nabla_{\ve{x}}\psi_0(\ve{x})+\nabla_{\ve{y}}\psi_1(\ve{x},\dfrac{\ve{x}}{l})+l\nabla_{\ve{x}}\psi_1(\ve{x},\dfrac{\ve{x}}{l})\big)\dm\ve{x}\\
&=\int_{\Gamma_n}g(\ve{x})\big(\psi_0(\ve{x})+l\psi_1(\ve{x},\dfrac{\ve{x}}{l})\big)\dm S_{\ve{x}}+\int_{\Omega}\psi_0(\ve{x})\dfrac{\tilde{\rho}_l(\ve{x})}{l}\dm\ve{x}+\int_{\Omega}\psi_1(\ve{x},\dfrac{\ve{x}}{l})\tilde{\rho}_l(\ve{x})\dm\ve{x}.
    \end{split}
\end{equation}

Taking the limit as $l \to 0$,
\bml{&\int_{\Omega\times\Box}\big(\nabla_{\ve{x}}\Phi_0(\ve{x})+\nabla_{\ve{y}}\Phi_1(\ve{x},\ve{y})\big)\cdot\big(\nabla_{\ve{x}}\psi_0(\ve{x})+\nabla_{\ve{y}}\psi_1(\ve{x},\ve{y})\big)\dm\ve{x}\dm\ve{y}\\
&=\int_{\Gamma_n}\psi_0(\ve{x})g(\ve{x})\dm S_{\ve{x}}+\lim_{l\to 0}\bigg[\int_{\Omega}\psi_0(\ve{x})\dfrac{\tilde{\rho}_l(\ve{x})}{l}\dm\ve{x}+\int_{\Omega}\psi_1(\ve{x},\dfrac{\ve{x}}{l})\tilde{\rho}_l(\ve{x})\dm\ve{x}\bigg],}
where $\Phi_1(\ve{x},\ve{y})\in~L^2[\Omega,H^1_{\#}(\Box)/\Scr{R}]$ \cite{allaire1992homogenization}. See Appendix \ref{Lax-Milgram} for the justification for the 2-scale limit.

Utilizing the weak limit of the second term on the right-hand side, as proven in Section \ref{2sTm}, and incorporating the two-scale limit of the final term, we derive the following:
\bml{&\int_{\Omega\times\Box}\big(\nabla_{\ve{x}}\Phi_0(\ve{x})+\nabla_{\ve{y}}\Phi_1(\ve{x},\ve{y})\big)\cdot\big(\nabla_{\ve{x}}\psi_0(\ve{x})+\nabla_{\ve{y}}\psi_1(\ve{x},\ve{y})\big)\dm\ve{x}\dm\ve{y}\\
&=\int_{\Gamma_n}\psi_0(\ve{x})g(\ve{x})\dm S_{\ve{x}}+\int_{\Omega}\nabla_{\ve{x}}\psi_0(\ve{x}).\ve{p}_0(\ve{x}) \dm\ve{x}\\
&+
\int_{\partial\Omega}\psi_0(\ve{x})\sigma(\ve{x})\dm S_{\ve{x}}+\int_{\Omega\times\Box}\psi_1(\ve{x},\ve{y})\tilde{\rho}_0(\ve{x},\ve{y})\dm\ve{y}\dm\ve{x},\label{WeakForm}}
where the two scale function $\tilde{\rho}_0(\ve{x},\ve{y})$, polarization $\ve{p}_0$ and surface charge $\sigma$ are as defined earlier.

If we set $\psi_0(\ve{x})=0$ in \eqref{WeakForm},
\beol{\int_{\Omega\times\Box}[\nabla_{\ve{x}}\Phi_0(\ve{x})+\nabla_{\ve{y}}\Phi_1(\ve{x},\ve{y})]\cdot\nabla_{\ve{y}}\psi_1(\ve{x},\ve{y})\dm\ve{y}\dm\ve{x}=\int_{\Omega\times\Box}\psi_1(\ve{x},\ve{y})\tilde{\rho}_0(\ve{x},\ve{y})\dm\ve{y}\dm\ve{x},}
performing integration by parts,
\bml{&\int_{\Omega\times\Box}\nabla_{\ve{y}}\cdot\big(\psi_1[\nabla_{\ve{x}}\Phi_0+\nabla_{\ve{y}}\Phi_1]\big)\dm\ve{y}\dm\ve{x}=\int_{\Omega\times\Box}\psi_1\Delta_{\ve{y}}\Phi_1\dm\ve{y}\dm\ve{x}+\int_{\Omega\times\Box}\psi_1\tilde{\rho}_0\dm\ve{y}\dm\ve{x},}
where we utilized the property that $\nabla_{\ve{y}} \cdot \nabla_{\ve{x}} \Phi_0(\ve{x}) = 0$. Furthermore, employing the result that the integral of the divergence of a periodic function over its period is zero, we obtain:
\beol{\int_{\Omega\times\Box}\big[\Delta_{\ve{y}}\Phi_1(\ve{x},\ve{y})+\tilde{\rho}_0(\ve{x},\ve{y})\Big]\psi_1(\ve{x},\ve{y})\dm\ve{y}\dm\ve{x}=0 ~~\forall\psi_1(\ve{x},\ve{y})\in L^2[\Omega,H^1_{\#}(\Box)/\Scr{R}].}
Using the density of $\psi_1$, we obtain:
\bml{&-\Delta_{\ve{y}}\Phi_1(\ve{x},\ve{y})=\tilde{\rho}_0(\ve{x},\ve{y})~\text{in}~\Omega\times\Box\\
&~~~\Phi_1(\ve{x},\ve{y}) \in~~L^2[\Omega;H^1_{\#}(\Box)/\Scr{R}].\label{Poisson}}

We now set $\psi_1(\ve{x},\ve{y})=0$ in \eqref{WeakForm},
\beol{\int_{\Omega\times\Box}(\nabla_{\ve{x}}\Phi_0+\nabla_{\ve{y}}\Phi_1)\cdot\nabla_{\ve{x}}\psi_0\dm\ve{y}\dm\ve{x}=\int_{\Gamma_n}\psi_0\big(g+\sigma\big)\dm S_{\ve{x}}+\int_{\Omega}\nabla_{\ve{x}}\psi_0\cdot \ve{p}_0\dm\ve{x},}
which, upon simplification, gives us
\bmlnn{&\int_{\Omega}\nabla_{\ve{x}}\Phi_0(\ve{x})\cdot\nabla_{\ve{x}}\psi_0(\ve{x})\dm\ve{x}+\int_{\Omega}\bigg[\int_{\Box}\nabla_{\ve{y}}\Phi_1(\ve{x},\ve{y})\dm\ve{y}\bigg]\cdot\nabla_{\ve{x}}\psi_0(\ve{x}) 
\dm\ve{x}\\
=&\int_{\Gamma_n}\psi_0(\ve{x})\big(g(\ve{x})+\sigma(\ve{x})\big)\dm S_{\ve{x}}+\int_{\partial\Omega}\psi_0(\ve{x})\ve{p}_0(\ve{x})\cdot\ve{n} \dm S_{\ve{x}}-\int_{\Omega}\psi_0(\ve{x})\divergence \ve{p}_0\dm\ve{x},\\
\implies&\int_{\Omega}\nabla_{\ve{x}}\Phi_0\cdot\nabla_{\ve{x}}\psi_0\dm\ve{x}=\int_{\Gamma_n}\psi_0(\ve{x})\big(g(\ve{x})+\sigma(\ve{x})+\ve{p}_0\cdot\ve{n}\big)\dm S_{\ve{x}}-\int_{\Omega}\psi_0\divergence\ve{p}_0\dm\ve{x},}
where we used the fact that $\nabla_{\ve{y}}\Phi_1$ has $0$-mean on $\Box$. Simplifying the above,
\beol{\int_{\Omega}\Big[-\Delta\Phi_0+\divergence \ve{p}_0\Big]\psi_0\dm\ve{x}=\int_{\Gamma_n}\psi_0(\ve{x})\big(-\dfrac{\partial\Phi_0}{\partial n}+g(\ve{x})+\sigma(\ve{x})+\ve{p}_0\cdot\ve{n}\big)\dm S_{\ve{x}}.}

Using the density of $\psi_0$,
\bml{\Delta_{\ve{x}}&\Phi_0(\ve{x})=\divergence\ve{p}_0~~\text{in}~\Omega\\
	 &\Phi_0=f~\text{on}~\Gamma_d\\
	 \dfrac{\partial}{\partial n}&\Phi_0=g+\ve{p}_0\cdot \ve{n}+\sigma~\text{on}~\Gamma_n.\label{Soln}}

\begin{rmrk}
{A standard equation in electrostatics is defined on $\Scr{R}^3$, as the potential extends throughout the entire space. Consequently, the previously derived strong form is the more physically meaningful representation. The motivation for deriving the strong forms for the ``interior problem" lies in its similarity to mechanics problems. Specifically, the problem could be reformulated with the charge replaced by a force, yielding an analogous limit equation for the homogenized displacement resulting from torque. This equivalence provides a broader range of physical phenomena to draw intuition from when analyzing specific examples. To develop intuition for the homogenization process, the ``interior" problem is utilized in \cite{sen2024nonuniqueness}, Section 2, to homogenize a one-dimensional problem.
}
\end{rmrk}

\begin{rmrk}
{In the homogenized limit, the macro-scale and micro-scale equations are decoupled, resulting in the corresponding potentials being independent. Two-scale convergence \cite{allaire1992homogenization} provides a framework to approximate the potential for finite $l$ as follows:
\beol{
\Phi_l(\ve{x}) = \Phi_0(\ve{x}) + l\Phi_1\left(\ve{x}, \frac{\ve{x}}{l}\right), \label{2-scale approx}
}
where it is assumed that $l$ is sufficiently small.
}
\end{rmrk}

\section{Polarization Theorem}\label{2sTm}
In this section, we rigorously establish the polarization theorem referenced during the derivation of the weak limit in Section \ref{WF}.

\begin{theorem}\label{PolThm}
Let $\{\tilde{\rho}_l\}_{l \in \scr{I}} : \Omega \to \Scr{R}$ denote a family of charge distributions satisfying the following constraints on the charge density:
\bml{
||\tilde{\rho}_l||_{L^2(\hat{\ve{x}}_l \oplus l\Box)} &:= \bigg[\int_{\Box} \tilde{\rho}^2_l(\hat{\ve{x}} + l\ve{y}) \, \dm\ve{y}\bigg]^{\frac{1}{2}} \leq C, \quad \forall \ve{x} \in \Omega,\\
\int_{\hat{\ve{x}}_l \oplus l\Box} \tilde{\rho}_l(\ve{x}) \, \dm\ve{x} &= 0, \qquad \forall \hat{\ve{x}}_l \oplus l\Box \subset \Omega,~\forall l \in \scr{I},
}
along with the following scaling constraints on the domain $\Omega$:
\bml{
\#(\Omega(\Box)) &= C\dfrac{|\Omega|}{l^3} + o(1), \\
\#(\Omega(\lr)) &= C\dfrac{|\partial\Omega|}{l^2} + o(1),\label{good}
}
where $o(1)$ denotes an error term that approaches zero as $l \to 0$. These conditions are assumed to hold for any arbitrary choice of unit cell.

Under the above assumptions, it can be established that, for all $\varphi \in D(\overline{\Omega})$, the following limiting relation holds:
\beol{
\lim_{l \to 0} \int_{\Omega} \varphi(\ve{x}) \dfrac{\tilde{\rho}_l(\ve{x})}{l} \, \dm\ve{x} =
\int_{\Omega \times \Box} \nabla_x \varphi(\ve{x}) \cdot \ve{y} \rho_0(\ve{x}, \ve{y}) \, \dm\ve{y} \dm\ve{x} +
\int_{\partial\Omega} \sigma(\ve{x}) \varphi(\ve{x}) \, \dm\ve{x}, \label{2ScaleThm}
}
where $\rho_0 \in L^2(\Omega \times \Box)$ represents the two-scale limit of the sequence of charge distributions, and $\sigma \in L^1(\partial\Omega)$ denotes the surface charge distribution arising from the partial unit cells.
\end{theorem}

The proof of Theorem \ref{PolThm} is divided into two parts. The first part addresses the bulk contribution in \eqref{2ScaleThm}, while the second part focuses on the surface contribution. The bulk contribution is analyzed in Section \ref{VT}, whereas the surface contribution is examined in Section \ref{SurfCrgDnsty}.

To establish the proof, we begin by incorporating the lattice $l\scr{L}$ into \eqref{2ScaleThm} and decomposing the integral into a sum over the corner map and integrals over individual unit cells. This yields:
\beol{
\sum_{\hat{\ve{x}} \in \Omega(\Box)} \int_{l\Box} \varphi(\hat{\ve{x}} + \ve{z}) \dfrac{1}{l} \tilde{\rho}_l(\hat{\ve{x}} + \ve{z}) \, \dm\ve{z} + \scr{R}_l(\Omega) =
\sum_{\hat{\ve{x}} \in \Omega(\Box)} l^3 \int_{\Box} \varphi(\hat{\ve{x}} + l\ve{y}) \dfrac{\tilde{\rho}_l(\hat{\ve{x}}, \ve{y})}{l} \, \dm\ve{y} + \scr{R}_l(\Omega),\label{(5.4)}
}
where $\scr{R}_l(\Omega)$ represents the residue term arising from the contributions of partial unit cells, which will subsequently give rise to the surface terms.

\subsection{Volume term}\label{VT}
In this section, we rigorously demonstrate that the first term on the right-hand side of \eqref{(5.4)} converges to the corresponding first term on the right-hand side of \eqref{2ScaleThm}.

\begin{proof}
We concentrate on the first summation term above,
\bml{&\sum_{\hat{\ve{x}}}l^3\Bigg[\int_{\Box}\underbrace{\dfrac{\varphi(\hat{\ve{x}}+l\ve{y})-\varphi(\hat{\ve{x}})-l\ve{y}\cdot\nabla\varphi(\hat{\ve{x}})}{l|\ve{y}|}}_{u(l,\ve{x})}|\ve{y}|\tilde{\rho}_l(\hat{\ve{x}},\ve{y})\dm\ve{y}\Bigg]\\
&+\sum_{\hat{\ve{x}}}l^3\Bigg[\dfrac{\varphi(\hat{\ve{x}})}{l}\underbrace{\int_{\Box}\tilde{\rho}_l(\hat{\ve{x}},\ve{y})\dm\ve{y}}_{=0}+\nabla\varphi(\hat{\ve{x}})\cdot\underbrace{\int_{\Box}\ve{y}\tilde{\rho}_l(\hat{\ve{x}},\ve{y})\dm\ve{y}}_{\ve{p}_l(\ve{x})}\Bigg],\label{(5.5)}}
where the second term vanishes due to the fact that $\tilde{\rho}_l$ has zero mean over $\Box$.

We now proceed to demonstrate that the first term \eqref{(5.5)} also converges to zero.
\bml{&\bigg\vert\sum_{\hat{\ve{x}}}l^3\int_{\Box}u(l,\ve{x})|\ve{y}|\tilde{\rho}(\hat{\ve{x}},\ve{y})\dm\ve{y}\bigg\vert\leq\sum_{\hat{\ve{x}}}l^3\int_{\Box}\bigg\vert u(l,\ve{x})|\ve{y}|\tilde{\rho}(\hat{\ve{x}},\ve{y})\bigg\vert \dm\ve{y}\\
&\leq\big[\sup_{\ve{x}\in\Omega}|u(l,\ve{x})|\big]\sum_{\hat{\ve{x}}}l^3\bigg|\bigg|\tilde{\rho}_l(\hat{\ve{x}},\ve{y})|\ve{y}|\bigg|\bigg|_{L^1(\Box)}\\
&\leq[\sup_{\ve{x}\in\Omega}|u(l,\ve{x})|]\sum_{\hat{\ve{x}}}l^3\bigg|\bigg||\ve{y}|\bigg|\bigg|_{L^2(\Box)}\bigg|\bigg|\tilde{\rho}_l(\hat{\ve{x}},\ve{y})\bigg|\bigg|_{L^2(\Box)}\\
&\leq C[\sup_{\ve{x}\in\Omega}|u(l,\ve{x})|]\sum_{\hat{\ve{x}}}l^3\underbrace{\big|\big|\tilde{\rho}_l(\hat{\ve{x}},\ve{y})\big|\big|_{L^2(\Box)}}_{C}\leq C\big[\sup_{\ve{x}\in\Omega}|u(l,\ve{x})|\big],\label{(5.6later)}}
where Hölder's inequality and the boundedness of $|\ve{y}|$ have been utilized to obtain the final expression. Since $\lim_{l \to 0} u(l, \ve{x}) \to 0$, as a consequence of $\varphi\in D$, we deduce that
\beol{\lim_{l\to 0}\bigg\vert\sum_{\hat{\ve{x}}}l^3\int_{\Box}u(l,\ve{x})|\ve{y}|\tilde{\rho}(\hat{\ve{x}},\ve{y})\dm\ve{y}\bigg\vert\leq C[\sup_{\ve{x}\in\Omega}\lim_{l\to 0}|u(l,\ve{x})|]=0.}

We arrive at the expression $\sum_{\hat{\ve{x}}} l^3 \nabla \varphi(\hat{\ve{x}}) \cdot \ve{p}_l(\hat{\ve{x}})$, which can be recognized as a sequence of Riemann sums. To rigorously establish the convergence of this Riemann sum, we proceed by employing a diagonal sequence argument as detailed below. First, we note that, due to the boundedness condition $||\tilde{\rho}_l||_{L^2(\Omega)} < \infty$ (see Theorem \eqref{GlobalBound}), there exists a weakly convergent subsequence. Furthermore, given the surjective relationship between $\tilde{\rho}_l$ and $\ve{p}_l$, the sequence $\ve{p}_l$ admits a weakly convergent subsequence. Since $\ve{p}_l$ converges weakly to $\ve{p}_0$, we obtain the following limit:
\beol{\lim_{l\to 0}\int_{\Omega}(\ve{p}_l(\ve{x})-\ve{p}_0(\ve{x}))\varphi(\ve{x})\dm\ve{x}=0.\label{weak}}

To establish that the Riemann sum converges to an integral, we proceed by proving the following lemma.

\begin{lemma}\label{Diag}
For any given $\ve{f}:\Omega \to \Scr{R}^3$ that is integrable, and $\ve{p}_l:\Omega \to \Scr{R}^3$ satisfying the boundedness condition $||\ve{p}_l||_{L^2(\Omega)} \leq C$ for some $C > 0$ and for all $l$, the following result holds:
\beol{\lim_{l\to 0}\sum_{\hat{\ve{x}}\in\Omega(\Box)}l^3 \ve{f}(\hat{\ve{x}})\cdot \ve{p}_l(\hat{\ve{x}})=\int_{\Omega}\ve{f}(\ve{x})\cdot \ve{p}_0(\ve{x})\dm\ve{x} }
\label{Riemann_weak}
\end{lemma}
\begin{proof}
We rewrite the left-hand side as 
\beol{I(l,d)=\sum_{\hat{\ve{x}}\in\Omega(\Box)}l^3\ve{f}(\hat{\ve{x}})\cdot \ve{p}_d(\hat{\ve{x}})}
We are interested in $I(l,l)$. We know that for a fixed $d$, we have the following limit,
\beol{\lim_{l\to 0}I(l,d)=\lim_{l\to 0}\sum_{\hat{\ve{x}}\in\Omega(\Box)}l^3 \ve{f}(\hat{\ve{x}})\cdot \ve{p}_d(\hat{\ve{x}})=\int_{\Omega}\ve{f}(\ve{x})\cdot \ve{p}_d(\ve{x})\dm\ve{x},\label{LIMIT}}
where we used the fact that Riemann integrals on bounded domains converge if the set of discontinuities has Lebesgue measure $0$. Since $\tilde{\rho}$ is locally periodic, $\ve{p}$ is continuous. The above expression can be succinctly represented as $I(\scr{I}, \{d\})$, where the limit in $l$ is taken over $\scr{I}$.

We construct a sequence for $d$ based on a convergent sequence for $l$. Let $\scr{I}$ denote a sequence of $l$ converging to $0$, and let the first term of this sequence be denoted by $\scr{I}(1)$. Since \eqref{LIMIT} holds for $\scr{I}$ and for all $d$, we apply \eqref{LIMIT} to $I(\scr{I}, \{\scr{I}(1)\})$ and extract a convergent subsequence, denoted $\scr{I}_1$. This process is iterated by extracting a subsequence $\scr{I}_2$ from $I(\scr{I}_1, \{\scr{I}_1(1)\})$, and continuing in this manner. Finally, we select the diagonal sequence, which ensures convergence along $I(\scr{I}, \scr{I})$.

\end{proof}

Hence, by applying the preceding theorem, we obtain
\beol{\lim_{l\to 0}\sum_{\hat{\ve{x}}\in\Omega(\Box)}l^3\nabla\varphi(\hat{\ve{x}}) \cdot \ve{p}_l(\hat{\ve{x}})=\int_{\Omega}\nabla\varphi(\ve{x})\cdot \ve{p}_{0}(\ve{x})\dm\ve{x}\label{Vol_Sum_weak} .}    
\end{proof}

\begin{rmrk}
{We make the following observation regarding the charge neutrality condition employed above. Imposing that the function be locally periodic and charge-neutral on each unit cell may prove excessively restrictive. Specifically, such a constraint could lead to either the set of admissible functions being periodic or the function itself dictating the choice of the lattice, thereby precluding arbitrary lattice shifts from yielding charge-neutral unit cells. A more flexible approach involves relaxing the charge neutrality condition to allow for charge neutrality with an error of $O(l^{\alpha})$, where the parameter $\alpha$ must be sufficiently large to ensure that the total charge in the domain vanishes in the homogenized limit. This relaxation offers the advantage that, even when the function is locally periodic, charge neutrality with an error of $O(l^{\alpha})$ can be enforced for any given choice of the lattice.
}
\end{rmrk}

\subsection{Surface Charge Density}\label{SurfCrgDnsty}

In this section, we rigorously demonstrate that the second term on the right-hand side of \eqref{(5.4)} converges to the corresponding second term on the right-hand side of \eqref{2ScaleThm}. 
Since we need to extract one convergent subsequence for \eqref{2ScaleThm}, we extract the surface convergent subsequence from that corresponding to the bulk subsequence in \eqref{Vol_Sum_weak}.

\begin{proof}

The residue term in \eqref{(5.4)} can be written as
\bml{\scr{R}_l(\Omega)&=\sum_{\hat{\ve{x}}\in\Omega(\lr)}\int_{\lr}\dfrac{\tilde{\rho}_l(\hat{\ve{x}},\ve{y})}{l}\varphi(\hat{\ve{x}}+l\ve{y})l^3\dm\ve{y},\\
					&=\underbrace{\sum_{\hat{\ve{x}}}l^2\int_{\lr}\tilde{\rho}_l(\hat{\ve{x}},\ve{y})\underbrace{[\varphi(\hat{\ve{x}}+l\ve{y})-\varphi(\hat{\ve{x}})]}_{u(l,\hat{\ve{x}},\ve{y})} \dm\ve{y}}_{\text{Term I}}+\underbrace{\sum_{\hat{\ve{x}}}l^2\varphi(\hat{\ve{x}})\underbrace{\int_{\lr}\tilde{\rho}_l(\hat{\ve{x}},\ve{y})\dm\ve{y}}_{\tilde{\sigma}_l(\hat{\ve{x}})}}_{\text{Term II}}.\label{Residue2}}

\rmrk{In the context of \eqref{Residue2}, the residue term was expressed as a summation over corner maps belonging to $\Omega(\lr)$. However, the definition of $\hat{\ve{x}} \in \Omega(\lr)$ does not necessarily ensure that $\hat{\ve{x}} \in \Omega$. This does not pose an issue for \eqref{Prob0}, as $\varphi$ is defined on $\Scr{R}^3$. However, it presents a challenge for \eqref{Prob}, since $\varphi$ is not defined outside the domain $\Omega$, i.e., on $\Omega^c$. 

To address this issue, we can redefine $\hat{\ve{x}}$ by replacing it with the point closest to the boundary of $\Omega$. For clarity and consistency, we will assume this adjustment has been implemented in all subsequent calculations.
}

Concentrating on the first term in \eqref{Residue2}, we proceed by applying the Cauchy-Schwarz inequality.
\bml{\text{Term I}&=\sum_{\hat{\ve{x}}}l^2|\tilde{\rho}(\hat{\ve{x}},\cdot)u(l,\hat{\ve{x}},\cdot)|_{L^1(\lr)}\leq\sum_{\hat{\ve{x}}}l^2||\tilde{\rho}_l(\hat{\ve{x}},\cdot)||_{L^2(\lr)}||u(l,\hat{\ve{x}},\cdot)||_{L^2(\lr)}\\
&\leq\underbrace{\sup_{x\in\Omega}||u(l,\hat{\ve{x}},\cdot)||_{L^2(\lr)}}_{u_{\#}(l)}\sum_{\hat{\ve{x}}}l^2\underbrace{||\tilde{\rho}_l(\hat{\ve{x}},\cdot)||_{L^2(\lr)}}_{\leq C}\leq u_{\#}(l)C\underbrace{\sum_{\hat{\ve{x}}}l^2}_{\#(\Omega(\lr))l^2}\\
&=u_{\#}(l)C+o(1)\label{(5.14later)},} 
where \eqref{good} has been employed to obtain the final result. Additionally, we utilized the observation that either $||\tilde{\rho}_l||_{L^2(\lr)} < 1$, in which case the upper bound is directly established, or alternatively, 
\[
1<||\tilde{\rho}_l||_{L^2(\lr)} < ||\tilde{\rho}_l||^2_{L^2(\lr)} < ||\tilde{\rho}_l||^2_{L^2(\Box)} < C,
\]
since the latter follows from the fact that measures on subsets are bounded above by measures on supersets.

Taking the limit as $l \to 0$ and invoking the continuity of the test function $\varphi$, we obtain
\beol{\lim_{l\to 0}\text{Term I}\leq C\lim_{l\to 0}u_{\#}(l)=0.}

We now proceed with the analysis of the second term in \eqref{Residue2}.
\beol{\sum_{\hat{\ve{x}}}l^2\varphi(\hat{\ve{x}})\tilde{\sigma}_l(\hat{\ve{x}})=\sum_{\hat{\ve{x}}}K_ll^2\varphi(\hat{\ve{x}})\underbrace{\dfrac{\tilde{\sigma}_l(\hat{\ve{x}})}{K_l}}_{{\sigma}_l(\hat{\ve{x}})}=\sum_{\hat{\ve{x}}}|S_l(\hat{\ve{x}})|\varphi(\hat{\ve{x}}){\sigma}_l(\hat{\ve{x}}),\label{Main part}}
where $|S_l(\hat{\ve{x}})|$ exhibits behavior akin to a surface measure, thereby ensuring that the limit of the expression above converges to a surface integral term.

Careful consideration is required when defining $K_l$ and $S_l$. For each partial unit cell, if the unit cell is connected, we define $S_l(\hat{\ve{x}}) := \partial\Omega \cap (\hat{\ve{x}} \oplus l\lr)$. However, in cases where the unit cell consists of disconnected components, a more nuanced approach is necessary. Specifically, if a disconnected component is entirely contained within the domain $\Omega$, we define $S_l(\hat{\ve{x}}) := \Omega \cap (\hat{\ve{x}} \oplus l\partial\lr)$. Conversely, if the component intersects the boundary, we define $S_l(\hat{\ve{x}})$ as $S_l(\hat{\ve{x}}) := \partial\Omega \cap (\hat{\ve{x}} \oplus l\lr)$.

In a more general scenario involving disconnected unit cells, there may be multiple disconnected components, some of which lie entirely within the domain $\Omega$, while others intersect the boundary $\partial\Omega$. Components located outside the domain do not contribute to the surface integral and can be disregarded. The surface measure $S_l$ must be defined individually for each disconnected component, with the overall $S_l$ representing the cumulative contribution from all relevant components. The following representation provides a concise formulation for calculating the surface charge:
\beol{\sigma_l(\hat{\ve{x}})=\dfrac{1}{K_l(\hat{\ve{x}})}\int_{\lr}\tilde{\rho}(\hat{\ve{x}},\ve{y})\dm\ve{y},\label{(5.17)}}
We define $K_l(\hat{\ve{x}})$ as $K_l(\hat{\ve{x}}) = \frac{S_l(\hat{\ve{x}})}{l^2}$. It is important to note that this definition accommodates partial unit cells that do not intersect with the boundary $\partial\Omega$. To ensure the validity of this definition, it remains to demonstrate that the corresponding surface charge is well-defined.

To extract a well-defined surface integral term, it is necessary to establish that the surface charge is bounded in $L^1(\partial\Omega)$. The surface charge, denoted by \eqref{(5.17)}, is defined as
\beol{\sigma_l(\ve{x})=\sum_{\hat{\ve{x}}\in\Omega(\lr)}\mathbbm{1}_{S_l}(\ve{x})\dfrac{\tilde{\sigma}_l(\hat{\ve{x}})}{K_l(\hat{\ve{x}})}\label{weaksurf},}
where $\tilde{\sigma}_l$ has been defined in \eqref{Residue2}. We now write
\bml{\int_{\partial\Omega}|\sigma_l(\ve{x})|\dm S_x&=\int_{\partial\Omega}\bigg|\sum_{\hat{\ve{x}}\in\Omega(\lr)}\mathbbm{1}_{S_l}(\ve{x})\dfrac{\tilde{\sigma}_l(\hat{\ve{x}})}{K_l(\hat{\ve{x}})}\bigg|\dm S_x\leq \int_{\partial\Omega}\sum_{\hat{\ve{x}}\in\Omega(\lr)}\mathbbm{1}_{S_l}(\ve{x})\dfrac{|\tilde{\sigma}_l(\hat{\ve{x}})|}{K_l(\hat{\ve{x}})}\dm S_x\\
&=\sum_{\hat{\ve{x}}\in\Omega(\lr)}\bigg[\int_{\partial\Omega}\mathbbm{1}_{S_l}(\ve{x})\dm S_x\bigg]\dfrac{|\tilde{\sigma}_l(\hat{\ve{x}})|}{K_l(\hat{\ve{x}})}=\sum_{\hat{\ve{x}}\in\Omega(\lr)}l^2\underbrace{|\tilde{\sigma}_l(\hat{\ve{x}})|}_{\leq C}=C+o(1),\label{5.19}}
where, in the final line, we employed \eqref{good} and utilized the inequality 
\beol{
|\tilde{\sigma}_l(\hat{\ve{x}})| \leq \int_{\lr} |\tilde{\rho}_l(\hat{\ve{x}}, \ve{y})| \, \dm\ve{y} = ||\tilde{\rho}_l||_{L^2(\lr)}\sqrt{|\lr|} \leq C~\max\left\{1, ||\tilde{\rho}_l||_{L^2(\Box)}\right\} < \infty,
}
where we relied on the fact that $|\lr|$ is bounded. Additionally, either $||\tilde{\rho}_l||_{L^2(\lr)} < 1$, in which case the bound is directly satisfied, or
\[
||\tilde{\rho}_l||_{L^2(\lr)} < ||\tilde{\rho}_l||^2_{L^2(\lr)} < ||\tilde{\rho}_l||^2_{L^2(\Box)} < C,
\]
where the latter inequality is a consequence of measures on subsets being bounded above by measures on supersets. This establishes that $\sigma_l \in L^1(\partial\Omega)$.

At this stage, we have established the boundedness of the sequence $\{\sigma_l\}_{\scr{I}}$. However, this alone is insufficient to extract a weak limit in $L^1$. To address this limitation, we can associate the sequence $\sigma_l$ with a corresponding sequence in the space of signed Radon measures. By doing so, we can seek weak-* convergence within the space of measures to obtain the desired surface charge distribution (see \cite{folland1999real}, Corollary 7.18). It is worth noting that the dual space in this context is $\scr{C}_c$, the space of compactly supported continuous functions.

Taking the limit of Eqn. \eqref{Main part}, and employing an analysis analogous to the argument presented in Lemma \eqref{Diag}, we can extract the following surface integral term:
\beol{\lim_{l\to0}\scr{R}_l(\Omega)=\int_{\partial\Omega}\varphi(\ve{x})\sigma(\ve{x})\dm S_{\ve{x}},}
where $\sigma$ is the weak limit of $\sigma_l$ defined in Eqn. \eqref{weaksurf}.    
\end{proof}

\rmrk{It is important to emphasize that the homogenization approach reveals the inherent topological structure of the problem. Notably, the bulk and surface contributions manifest as distinct components, each playing a unique role in determining the effective behavior of the system.
}

\section{Non-uniqueness of the dipole moment}\label{Non-uniqueness}
The polarization, defined as the dipole density per unit volume, is an inherently ill-defined quantity. To illustrate this, we consider the following example:
\begin{exmp}\label{non-uniEx}
Consider a charge distribution $\rho_l : \Scr{R}^3 \to \Scr{R}$ defined as $\rho_l(\ve{x}) = \frac{1}{l} \sin\left(2\pi \frac{x_1}{l}\right)$, which is $l$-periodic. We now evaluate the dipole moment with respect to an arbitrary unit cell located at $\ve{a} \oplus l[0,1]^3$.
\beol{\ve{p}(\ve{a})=\int_{\ve{a}\oplus l[0,1]^3}\dfrac{(\ve{x}-\ve{a})}{l}\dfrac{\rho_l(\ve{x})}{l^3}\dm V_{\ve{x}}=\int_{[0,1]}\ve{y}\sin(2\pi\dfrac{a_1}{l}+2\pi y_1) \dm y_1=-\dfrac{1}{2\pi}\cos{(2\pi\dfrac{a_1}{l})}(1,0,0).}
Thus, depending on the choice of the unit cell (determined by the position of $\ve{a}$), different values for the polarization vector $\ve{p}$ are obtained. It is important to emphasize that while the total charge remains invariant, the polarization varies based on the selection of the unit cell.\qed
\end{exmp}

It is essential to address the issue of non-uniqueness in the polarization distribution arising from the choice of unit cell. The problem has been formulated such that the selection of the origin $O$, the lattice $\scr{L}$, and the shape of the unit cell $\Box$ are arbitrary. We denote this tuple by $(O, \scr{L}, \Box)$.

Consider a sequence of charge distributions. The polarization and total surface charge calculated for a given tuple $(O, \scr{L}, \Box)$ are denoted by $\{\ve{p}_l\} : \Omega(\Box) \to \Scr{R}^3$ and $\{\sigma_l + \ve{p}_l \cdot \ve{n}\} : \partial\Omega \to \Scr{R}$, respectively. If, instead, a different tuple $(\utilde{O}, \utilde{\scr{L}}, \utilde{\Box})$ were chosen, the corresponding sequence of polarization and total surface charge would be denoted by $\utilde{\ve{p}}_l : \Omega(\utilde{\Box}) \to \Scr{R}^3$ and $\{\utilde{\sigma}_l + \utilde{\ve{p}}_l \cdot \ve{n}\} : \partial\Omega \to \Scr{R}$, respectively.

From Example \ref{non-uniEx}, it can be demonstrated that for any $l$, $\ve{p}_l \neq \utilde{\ve{p}}_l$ and $\sigma_l \neq \utilde{\sigma}_l$. The discrepancy in the surface charge arises due to variations in the shape of the partial unit cells for different choices of unit cells. Consequently, since the polarization and surface charge distribution differ for arbitrary $l \in \scr{I}$, they will also differ in the limit. This represents a significant issue, as it implies that the homogenized equations being solved in \eqref{Soln0} and \eqref{Soln} depend on the choice of the unit cell. It is important to note, however, that the corrector equation remains unaffected, as it is formulated with periodic boundary conditions.

We assert that, although $\ve{p}_0$ and $\sigma$ are not unique, the potentials calculated via the homogenized problem \eqref{2-scale approx} are unique. To establish this claim, we first emphasize that the corrector equation remains invariant, as it is formulated with periodic boundary conditions. Thus, it suffices to demonstrate that the homogenized potential is also invariant. 

To begin, we note that the sequence of potentials $\{\Phi_l\}_{\scr{I}}$ converges weakly to $\Phi_0$. The potentials $\Phi_l$ satisfy Poisson's equation for the charge distribution $\{\rho_l\}_{\scr{I}}$ on $\Omega$. Since the convergent limits are extracted directly from the original charge distribution $\{\rho_l\}_{\scr{I}}$, the resulting homogenized equation, for a specific choice of lattice and unit cell, is unique. Consequently, the corresponding homogenized potential for this choice is also unique.

Now, consider a different choice of lattice and unit cell. The sequence of potentials $\{\Phi_l\}_{\scr{I}}$ would still satisfy Poisson's equation for the same charge distribution $\{\rho_l\}_{\scr{I}}$ under the same boundary conditions. Due to the uniqueness of weak convergence, the final homogenized potential remains unchanged, which implies that the homogenized potential is independent of the choice of unit cells. 

As $\Phi_0$ is unique, it follows from \eqref{Soln0} and \eqref{Soln} that $\divergence \ve{p}_0$ and $\sigma + \ve{p}_0 \cdot \ve{n}$ are also unique and independent of the choice of unit cells.


\section*{Acknowledgments}
    We thank NSF (DMS 2108784, DMS 2342349), AFOSR (MURI FA9550-18-1-0095), and ARO (MURI W911NF-24-2-0184) for financial support.
    Shoham Sen was also partly supported by a Vannevar Bush Faculty Fellowship at the University of Minnesota (PI: R. D. James).
    Kaushik Dayal thanks the Air Force Research Laboratory for hosting his visits.

\appendix
\section{Two-scale Convergence}
The two scale theorem \cite{allaire1992homogenization} is;
\begin{theorem}
A family $\{u_{\varepsilon}\}\subset L^2(\Omega)$ is said to two-scale converge to $u_0(x,y)\in L^2(\Omega\times\Box)$ if $\forall$ $\psi(x,y)\in D[\Omega;\scr{C}^{\infty}_{\#}(\Box)]$, we have
\beol{\lim_{\varepsilon\to 0}\int_{\Omega}u_{\varepsilon}(x)\psi(x,\dfrac{x}{\varepsilon})\dm x=\int_{\Omega\times\Box}u_0(x,y)\psi(x,y)\dm y \dm x.}
\end{theorem}

\begin{enumerate}
\item A compactness theorem states that if $||u_{\varepsilon}||_{L^2(\Omega)}<C~\forall\varepsilon$, then we can extract a subsequence that two-scale convergence.
\item If $u_{\varepsilon}\overset{2s}{\longrightarrow}u_0$ in $L^2$, then we have $u_{\varepsilon}\rightharpoonup\int_{\Box}u_0(\cdot,y)\dm y$.
Thus weak limit can be considered as a special case of two-scale limits.
\end{enumerate}

\section{Bounded in Sobolev Space}\label{Lax-Milgram}

We show below that the sequence of functions $\Phi_l$ are bounded in $H^1(\Omega)$, implying that we can extract a 2-scale limit. 
From Proposition 3.32 of \cite{cioranescu1999introduction}, we have for $f\in H^{\frac{1}{2}}(\Gamma_d)$, a function $F\in H^1(\Omega)$ that is constant with respect to the sequence $l$. 
Consider the following inequality:
\beol{||\Phi_l||_{H^1(\Omega)}\leq||\Phi_l-F||_{H^1(\Omega)}+||F||_{H^1(\Omega)}\leq ||z_l||_{L^2(\Omega)}+||\nabla z_l||_{L^2(\Omega)}+||F||_{H^1(\Omega)},}
where we used the Minkowski inequality and defined $z_l=\Phi_l-F$. Using that $z_l\vert_{\Gamma_d}=0$, we can apply the modified Poincare inequality (Proposition 3.36 in \cite{cioranescu1999introduction}) to get:
\beol{||\Phi_l||_{H^1(\Omega)}\leq(1+C_{\Omega}) ||\nabla z_l||_{L^2(\Omega)}+||F||_{H^1(\Omega)},}
where the above inequality states that if we can bound $\nabla z_l$ in $L^2(\Omega)$, then we can bound $\Phi_l$ in $H^1(\Omega)$.

Finally, consider:
    \beol{||\nabla z_l||_{L^2(\Omega)}=||\nabla\Phi_l-\nabla F||_{L^2(\Omega)}\leq ||\nabla\Phi_l||_{L^2(\Omega)}+||\nabla F||_{L^2(\Omega)},}
where we have used the Minkowski inequality. 
Since $\nabla \Phi_l$ is uniformly bounded in $l$, and $\nabla F$ is in $L^2(\Omega)$ by definition, we therefore have a bound for $\{\nabla z_l\}_{\scr{I}}$, giving us a bound for $\{\Phi_l\}_{\scr{I}}$ in $H^1(\Omega)$. This allows us to extract a 2-scale subsequence.

\newcommand{\etalchar}[1]{$^{#1}$}

\end{document}